\documentclass[centertags,12pt]{amsart}

\usepackage{amssymb}
\usepackage{amsmath}
\usepackage{bm}
\usepackage{amsthm}
\usepackage{xcolor}
\usepackage{graphicx}

\makeatletter
\renewcommand{\subsection}{\@startsection
{subsection}{2}{0mm}{\baselineskip}{-0.25cm}
{\normalfont\normalsize\em}}
\makeatother

\newtheorem{theorem}{Theorem}
\newtheorem{proposition}[theorem]{Proposition}

{\theoremstyle{definition}

\newtheorem{example}{Example}
\theoremstyle{remark}

\title{On Mignotte secret sharing schemes over Gaussian Integers}
\author{Diego Munuera-Merayo}
\date{\today}
}

\begin{document}
\maketitle

\begin{abstract}
Secret Sharing Schemes (SSS) are methods for distributing a secret among a set of participants.
One of the first Secret Sharing Schemes was proposed by M. Mignotte, based on the Chinese remainder theorem over the ring of integers. In this article we extend the Mignotte's scheme to the ring of Gaussian Integers and study some of its properties.
While doing this we aim to solve a gap in a previous construction of such extension. In addition we show that any access structure can be made through a SSS over $ \mathbb{Z}[i]$.

\end{abstract}

\section{Introduction}

Secret sharing schemes (SSS's) are methods for distributing a secret among a set of participants. Among their applications we may highlight key management, threshold and visual cryptography, e-voting, multiparty computation and more. Given the increasing importance of information security in today's society, it is not surprising that SSS's have become an active field of research within public key cryptography.

In a SSS, a secret $s$ is broken
into $n$ pieces or  {\em shares}. Each share is given to one out of  $n$ participants. This shares are designed ensuring that some authorized coalitions of participants can reconstruct the secret by pooling the shares of its members, while non-authorized coalitions cannot do it.  This family of authorized coalitions is called {\em access structure} of the scheme.

The problem of secret sharing was firstly introduced by A. Shamir in \cite{Shamir}. In that article he also suggested a method to perform it: the so-called Shamir threshold secret sharing schemes, based on Lagrangian interpolation. Since the publication of Shamir's work, several other methods have been proposed. For our purposes within this article, we may highlight the ones by M. Mignotte \cite{Mignotte} and C. Asmuth and J. Bloom \cite{AB}, both based on the Chinese remainder theorem for coprime modules over the ring of integers $\mathbb{Z}$.
This three methods belong to the category of  {\em threshold} $(t, n)$ schemes. This means that their access structures are formed by all the coalitions with at least $t$ out of $n$ participants, for a certain $t$.

On the other hand,
the Chinese remainder theorem has many known applications in cryptography, \cite{DAC}. It is based on a well known property of integers, the so-called {\em Euclidean division:} for any two integers $n,m$ with $m\neq 0$, there exists integers $q$ (quotient) and $r$ (remainder) such that $m=nq+r$ with $0\le r<|m|$.

Of course, Euclidean division is not exclusive to $\mathbb{Z}$. As there exist other rings that admit it, the called Euclidean Domains. Among these we may highlight the ring of univariate polynomials over a field  $\mathbb{K}[X]$, and the ring of Gaussian Integers $\mathbb{Z}[i]$. As a consequence a Chinese remainder theorem holds over both $\mathbb{K}[X]$ and $\mathbb{Z}[i]$, and therefore we can extend the Mignotte secret sharing scheme to these rings.

In the literature we can find successive generalizations of Mignotte's scheme. Firstly by S. Iftene to not necessarily coprime integers \cite{Iftene}. Later  T. Galibus and G. Matveev \cite{GMTodas} provided a version over polynomial rings of one variable. In that paper the authors also showed the remarkable property that any access structure can be realized  by a  Mignotte polynomial SSS.  Finally, in  \cite{OTS}   an extension of  Mignotte threshold scheme to the ring of Gaussian Integers is proposed, and subsequently applied to the problem of image sharing.Recall that Gaussian Integers play a role in digital signal processing, cryptography, coding, and many other fields of science and technology, see \cite{Blahut}. And that Mignotte's threshold scheme has already been proposed for image sharing, \cite{ImageSharing}.

However, the extension of Mignotte SSS to Gaussian Integers proposed in \cite{OTS} contains a  crucial gap. Contrary to what occurs over $\mathbb{Z}$ and $\mathbb{K}[X]$,
in the Euclidean division over  $\mathbb{Z}[i]$, quotient and remainder are not unique. Thus, given $v,z\in \mathbb{Z}[i]$, the expression $v \, (\mbox{mod } z)$ does not uniquely identify one element of $\mathbb{Z}[i]$. This causes that the solution of a system of simultaneous congruence equations is not unique over  $\mathbb{Z}[i]$ and, finally, that Mignotte's scheme based on it  may provide wrong recovered secrets.

In this article we propose a new version of Mignotte's SSS over $\mathbb{Z}[i]$ that works properly and investigate some of its properties. The paper is organized as follows:  in Section 2 we recall some basic facts on SSS's and the arithmetic of $\mathbb{Z}[i]$. The proposed scheme is developed in Section 3, where  we also study its main properties. In particular, following \cite{GMTodas} we show that any access structure can be realized by our scheme.  Also we characterize those access structures which are obtained from coprime modules.  Some examples are included in order to illustrate the obtained results.

\section{Some background on SSS's and Gaussian Integers}

First, we recall some known facts about secret sharing schemes and Gaussian Integers. A complete study of most of this subjects can be found in  \cite{Stinson} for secret sharing  and \cite{Fra} for Gaussian Integers.


\subsection{Secret sharing schemes}

Let ${\mathcal{P}} = \{1, \dots, n\} $ be a set of $n$  {\em participants} and ${\mathcal{S}}$ be a finite and non-empty  {\em set of secrets}. We want to share a secret $s \in {\mathcal{S}} $ among the participants in $\mathcal{P}$. For that purpose, each participant will receive a data $s_i$ about the secret, which we will call its {\em share}. A dealer computes the $s_1,\dots,s_n$ from $s$ and assigns $s_i$ to  each participant $i$.
A  {\em secret sharing scheme} is a method $ {\mathcal{R}}$ of calculating the shares $s_i $ so that certain groups of participants, previously determined, can recover $s$ by pooling the shares of their members; while making it impossible for any other coalition of participants to recover the secret. 

We will call {\em access structure} of the scheme (denoted by $\mathcal{A}$) to the family of all coalitions authorized to recover the secret. Note that this family is monotone increasing. 


We say that a scheme is {\em perfect} if unauthorized coalitions cannot deduce any information about the secret.
Another important feature to take into account in a scheme ${\mathcal {R}} $  is the size of the shares distributed to the participants.
Let us denote by ${\mathcal{S}}_ i$ the set of all possible values that $s_i$ can take when $s$
runs over ${\mathcal {S}}$. It can be proved that when no unauthorized coalition can discard any element of $\mathcal{S}$ to be the shared secret,  it holds that  $|{\mathcal{S}}_i|\ge |{\mathcal{S}}|$.
We define the {\em information rate} of $\mathcal{R}$ as
$$
\rho(\mathcal{R})=\min \left\{  \frac{\log|{\mathcal{S}}_i|}{\log|{\mathcal{S}}|} \ : \ i=1,\dots,n  \right\}.
$$
If $\rho(\mathcal{R})=1$ then the scheme $\mathcal{R}$ is called {\em ideal}.

\subsection{The Chinese remainder theorem}

In its classical version, over the ring of integers, the Chinese remainder theorem states the following.

\begin{theorem}
\label{chino}
Let $a_1,\dots,a_n \in \mathbb{Z}$,  and let $m_1,\dots,m_n$ be relatively pairwise coprime integers such that $m_1,\dots,m_n\ge 2$. The system of simultaneous congruence equations
$$
x\equiv a_i \; (\mbox{\rm mod } m_i), \hspace*{4mm} i=1,\dots,n
$$
has a unique solution modulo $\mbox{\rm lcm}(m_1,\dots,m_n)$. Furthermore, this solution can be explicitly given as
$
x=a_1q_1r_1+\cdots+a_nq_nr_n
$. Where $q_i=m_1\cdots m_n/m_i$ and  $r_i=q_i^{-1}$ in  $\mathbb{Z}/m_i\mathbb{Z}$, for $i=1,\dots,n$.
\end{theorem}

This theorem can be generalized, as the integers need not to be relatively pairwise coprime. In this case there exists a solution if and only if $a_{i} \equiv a_{j}  (\mbox{\rm mod } \mbox{ lcm}(m_{i}, m_{j})  ) $ for all $ 1 \leq i,j \leq n $. Furthermore, it can be stated over Euclidean rings, such as Gaussian Integers. 

\subsection{Mignotte SSS over $\mathbb{Z}$}

The Mignotte's original secret sharing scheme allows the construction of threshold $(t,n)$ schemes over a set $\mathcal{P}=\{ 1,\dots, n\}$ of $n$ participants as follows \cite{Mignotte}: Let $\mathbf{m}:m_1<m_2<\cdots<m_n$, be a sequence of $n$ pairwise relatively coprime integers such that $m_{n-t+2}\cdots m_n < m_1\cdots m_t$. The integer $m_i$ is assigned to participant $i$. Let us write $m^-=m_{n-t+2}\cdots m_n$, $ m^+=m_1\cdots m_t$ and for a coalition $C\subseteq \mathcal{P}$, let $m(C)=\prod_{i\in C} m_i$. The above condition on the $m_i$'s implies that $m(C)\le m^-$ when $|C|<t$ and $m(C)\ge m^+$ when $|C|\ge t$. 

The scheme operates as follows: Let $\mathcal{S}=\{ s\in \mathbb{Z} \ : \ m^-  < s < m^+\}$ be the set of secrets to be shared.
Given a secret $s$, the share of participant $i$ is $s_i= s \ (\mbox{\rm mod } m_i)$. An authorized coalition $A$ with $|A|\ge t$ may recover the secret by solving the system
$$
(S_A)  \hspace*{12mm} x\equiv s_i \; (\mbox{\rm mod } m_i), \hspace*{4mm} i\in A
$$
whose solution $x$ is unique modulo $m(A)$ as guaranteed by the Chinese remainder theorem. Since $s< m^+\le m(A)$, it holds that $x=s$.
An unauthorized coalition $B$ with $|B|< t$, may try to recover the secret by solving the system
$$
(S_B)  \hspace*{12mm} x\equiv s_i \; (\mbox{\rm mod } m_i), \hspace*{4mm} i\in B
$$
whose solution $x$ is unique modulo $m(B)$. But since $x<m(B)\le m^-\le s$, it holds that $x\neq s$ and so $B$ does not recover the legitimate secret.

\subsection{Extensions of Mignotte SSS}

In \cite{Iftene} Iftene suggested an extension of Mignotte SSS over the integers as follows. Let $\mathbf{m}:m_1,m_2,\cdots,m_n$, be a sequence of $n$ (not necessarily coprime)  integers. For a coalition $C\subseteq\mathcal{P}$ let us denote by $\mbox{\rm lcm}(C)$ the least common multiple of $\{ m_i  : i\in C\}$. Take two integers $m^-<m^+$ such that the interval $(m^-,m^+)$ does not contain the  $\mbox{\rm lcm}(C)$ of any coalition $C\subseteq\mathcal{P}$. Furthermore, we impose that $\frac{\pi (m^{+}- 4m^{-})}{4m^{-}} > 1$.
In this setting let us consider the family $\mathcal{A}=\mathcal{A}(\mathbf{m},m^+) = \{ A\subseteq \mathcal{P} : \mbox{lcm}(A)\ge m^+\}$.
This is a monotonous increasing family, so we may consider it as an access structure over $\mathcal{P}$. The set of secrets is $\mathcal{S}=\{ s\in \mathbb{Z} \ : \ m^-\le s < m^+\}$. Sharing and reconstruction of  secrets is carried out as in the original Mignotte scheme.



Another extension of Mignotte's scheme to the ring of polynomials  in one variable over a finite field, $\mathbb{F}_q[X]$, was given in \cite{GMTodas}. Following this idea, other extensions of the method to Euclidean rings have been made, such as \cite{OTS} or \cite{ImageSharing}. From now on we will focus on the extension to Gaussian Integers.

\subsection{The ring of Gaussian Integers}

A {\em Gaussian Integer} is a complex number $z=a+bi$, where both $a$ and $b$ are integers. The set of Gaussian Integers is thus $\mathbb{Z}[i]$. Over this ring we may consider the Euclidean function {\em norm} of $z$, $N(z)=a^2+b^2=z\overline{z}$ (where $\overline{z}$ denotes the complex conjugated of $z$). Thus the norm is multiplicative, $N(vz)=N(v)N(z)$ for all $v,z \in \mathbb{Z}[i]$ and satisfies $N(v)\le N(vz)$ if $z\neq 0$.

The most interesting arithmetic property of $\mathbb{Z}[i]$ is the existence of an Euclidean division with respect to the \textit{norm}: given $v,z\in\mathbb{Z}[i]$ with $z\neq 0$, there exist $q,r \in\mathbb{Z}[i]$ such that $v=zq+r$ with $N(r)<N(z)$. We say that $\mathbb{Z}[i]$ is an Euclidean Domain. Thus   $\mathbb{Z}[i]$  is also a Principal Ideal Domain \cite{Fra}, and therefore a Chinese remainder theorem similar to Theorem \ref{chino}, holds over it.

We observe, however, that the quotient and remainder of the Euclidean division are not
uniquely determined, unlike what we have over $\mathbb{F}_q[X]$ or what we can do over $\mathbb{Z}$, see \cite{Jod}. For example
we have
$10i = (5+4i)(1+i)+(i-1) = (5+4i)(1+2i)+(3-4i)$. Thus there is no a consistent way to define the expression $v\; (\mbox{mod}  \;z)$, as it does not uniquely identify any element. As a result, the solution 
of a system of congruence equations is not uniquely determined.


\subsection{The Mignotte  SSS  over $\mathbb{Z}[i]$ of \cite{OTS}}

In \cite{OTS} the authors propose an extension of Mignotte's scheme to a threshold $(t,n)$ SSS over $\mathbb{Z}[i]$, and they  use this extension to give a method for sharing secret images. However they did not take into account the lack of uniqueness of Euclidean division. This causes that the proposed SSS may lead to reconstruct the wrong secrets. 

The method goes as expected: In order to share a secret among a set $\mathcal{P}$ of $n$ participants, we begin from a sequence $\mathbf{m}: m_1,\dots,m_n$ of pairwise coprime Gaussian Integers such that $N(m_1)<\cdots<N(m_n)$ and
$N(m_{n-t+2}\cdots m_n)<N(m_1\cdots m_t)$. The set of secrets is $\mathcal{S}=\{ s\in \mathbb{Z}[i] : N(m_{n-t+2}\cdots m_n)\le  N(s) <N(m_1\cdots m_t) \}$. The sharing of a secret $s\in \mathcal{S}$ and its
subsequent reconstruction are  performed in the usual way: Given a secret $s$, the share of participant $i$ is $s_i= s \ (\mbox{\rm mod } m_i)$. A  coalition $A$ with $|A|\ge t$, may recover the secret by solving the system
$$
(S_A)  \hspace*{12mm} x\equiv s_i \; (\mbox{\rm mod } m_i), \hspace*{4mm} i\in A .
$$ 
We will now show an example of how it operates and how it may lead to mistakes.

\begin{example}
(Example 3.3 of \cite{OTS}).
Let $n=3$, $\mathbf{m}: 7+4i, -3-13i, 11+8i$ and take $t=2$. The set of secrets is then
$\mathcal{S}=\{ s\in \mathbb{Z}[i] : 185 \le N(s) <11570 \}$. Let $s=70-70i$. Note that $N(s)=9800$ so $s$ is a valid secret. The shares are $s_1=1+2i, s_2=4, s_3=3-i$. The authorized coalition $\{ 2,3\}$ wants
to recover the secret. To that end they solve the system
$$
\left\{ \begin{array}{lll}
x \equiv &  4  & (\mbox{\rm mod } -3-13i) \\
x \equiv &  3-i & (\mbox{\rm mod } 11+8i) \\
\end{array} \right.
$$
whose solution is  $x=70-70i$, but also $x=-1+97i$, and both are valid secrets. Most computer systems choose the solution of smaller norm. Since $N(-1+97i)=9410<N(70+70i) = 9800$, then the coalition $\{ 2,3\}$ recovers $s=-1+97i$, which is a wrong secret.
\end{example}

\section{A Mignotte SSS over Gaussian Integers}

We now aim to develop an extension of Mignotte SSS to Gaussian Integers, making sure that the remainders are uniquely determined. That is, to ensure that the expression $v \ (\mbox{\rm mod } z)$ refers to a unique element of $\mathbb{Z}[i]$. On top of that, we will not impose the condition that the modules $m_{1},\cdots, m_{n}$ are coprime, so that we find more general schemes.

\subsection{Fundamental domains}

In order to ensure the uniqueness of the remainder, we will operate in restrictions of $\mathbb{Z}[i]$. 
Let $z\in\mathbb{Z}[i]$, $z\neq 0$.We may define the {\em fundamental domain} of $z$ is the subset of $\mathbb{C}$
$$
\mathcal{F}(z)=\{ z(\alpha+\beta i)  \ : \ \alpha,\beta \in \mathbb{R},  \mbox{$-\frac12< \alpha, \beta\le \frac12 $}\}.
$$
Geometrically $\mathcal{F}(z)$ is a semi-open square in $\mathbb{C}\sim \mathbb{R}^2$, with vertices
$$
\mbox{$
z(\frac 12+\frac 12 i), \;
z(\frac 12-\frac 12 i), \;
z(-\frac 12-\frac 12 i), \,
z(-\frac 12+\frac 12 i).
$}
$$

\begin{center}

\includegraphics[width=5cm]{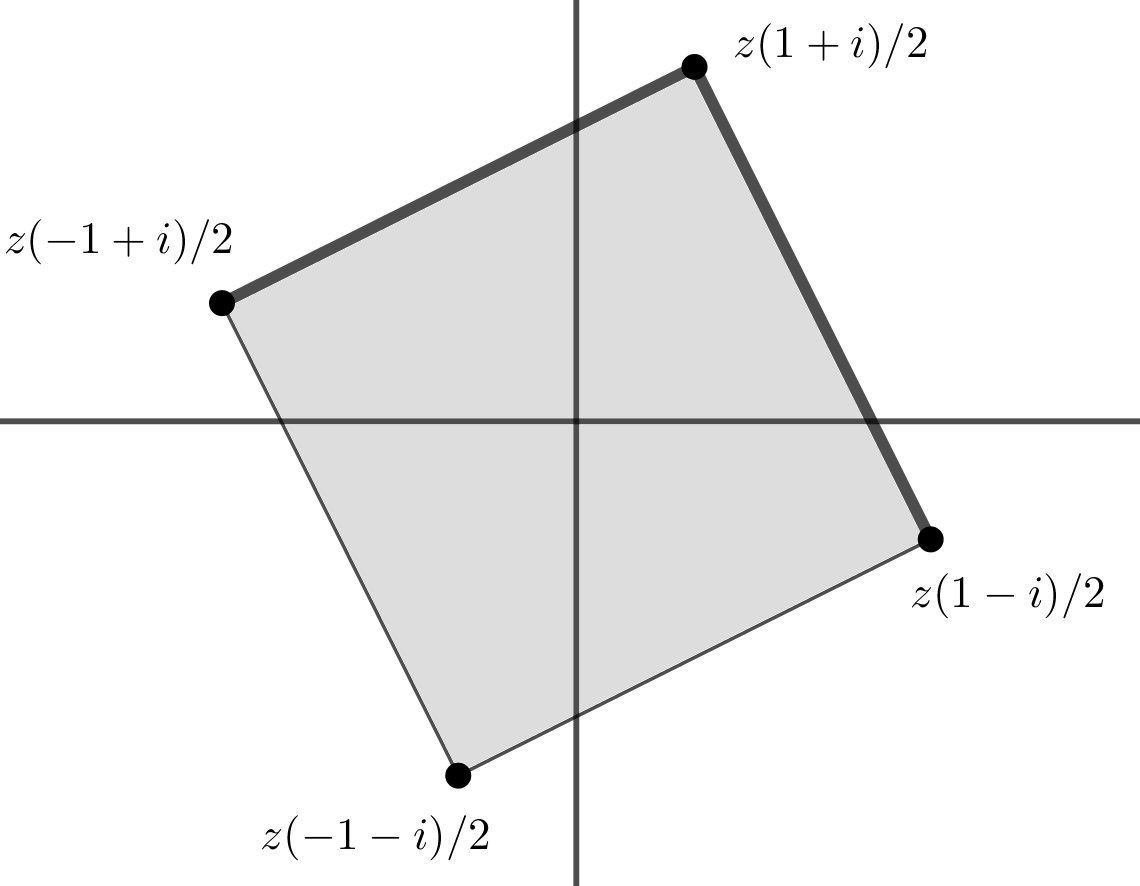}

\end{center}

As seen in the above figure $\mathcal{F}(z)$ is centered at 0 and has side length $|z|=\sqrt{N(z)}$ . Note that for every Gaussian Integer $v\in \mathcal{F}(z)$ it holds that $N(v) \le N(z)/2$. On this space, we find the following.

\begin{proposition}
\label{DivUnica}
Given two Gaussian Integers $v,z\in\mathbb{Z}[i]$ with $z\neq 0$, there exists  $q,r\in\mathbb{Z}[i]$ such that $r\in \mathcal{F}(z)$ and $v=zq+r$. Moreover such $q$ and $r$ are unique and  $N(r) \le N(z)/2$.
\end{proposition}
\begin{proof}
Let us consider the complex number $v/z=x+yi$. Let $a$, $b$  be the integers
$a=\lceil x -\frac{1}{2} \rceil$ and $b=\lceil y-\frac{1}{2} \rceil$. Thus $-\frac{1}{2}< x-a\le \frac{1}{2}$, $-\frac{1}{2}< y-b\le \frac{1}{2}$. Let $q=a+bi$. Then  $v=zq+(z(v/z-q))=zq+r$ where $r=z(v/z-q) = z((x-a)+(y-b)i ) \in \mathcal{F}(z)$. To see the uniqueness, if $zq_1+r_1=zq_2+r_2$ with $r_1,r_2\in \mathcal{F}(z)$, then
write $r_1=z(\alpha_1+\beta_1 i)$,  $r_2=z(\alpha_2+\beta_2 i)$ with $-\frac{1}{2}< \alpha_1, \beta_1, \alpha_2, \beta_2\le  \frac{1}{2}$. We have $z(q_2-q_1)=r_1-r_2=z((\alpha_1-\alpha_2)+(\beta_1-\beta_2) i)$, hence $\alpha_1-\alpha_2$ and $\beta_1-\beta_2$ are both integers. So $\alpha_1=\alpha_2$ and $\beta_1=\beta_2$.
The last statement, $N(r) \le N(z)/2$, is a consequence of the fact that $r\in \mathcal{F}(z)$.
\end{proof}

Given  $v,z\in\mathbb{Z}[i]$ with $z\neq 0$ as above, we say that the unique remainder $r\in \mathcal{F}(z)$ of the division $v=zq+r$, is the {\em principal value} of $v \ (\mbox{\rm mod } z)$. We will write this as
$r=v \ (\mbox{\bf mod } z)$. Let us note that the above proof gives a procedure to explicitly compute this number.

Besides $\mathcal{F}(z)$ we shall also consider the {\em strict fundamental domain} of $z$, which is the set $\mathcal{F}^o(z)$ of points in $\mathcal{F}(z)$ but  not on its border
$$
\mathcal{F}^{o}(z)=\{ z(\alpha+\beta i)  \ : \ \alpha,\beta \in \mathbb{R},  \mbox{$-\frac 12< \alpha, \beta< \frac12)$}\}
$$
and the {\em closure} of $\mathcal{F}(z)$
$$
\overline{\mathcal{F}}(z)=\{ z(\alpha+\beta i)  \ : \ \alpha,\beta \in \mathbb{R},  \mbox{$-\frac 12\le \alpha, \beta \le \frac12)$}\}.
$$
Clearly $\mathcal{F}^{o}(z)\subset \mathcal{F}(z)\subset\overline{\mathcal{F}}(z)$.
Now we will expose some properties of these domains, that will allow us to successfully develop the scheme. 
\begin{proposition}
\label{contenciones}
Let $v,z\in \mathbb{Z}[i]$, $z\neq 0$.\newline
(a) If $N(v)< N(z)/2$, then ${\mathcal{F}}(v)\subset {\mathcal{F}}(z)$.  \newline
(b) $\{ u\in \mathbb{Z}[i]  :  N(u)< N(z/2)\} \subset {\mathcal{F}}^o(z) \subset \{ u\in \mathbb{Z}[i]  :  N(u)\le N(z)/2\}$.
\end{proposition}
\begin{proof}
(a)  Since both ${\mathcal{F}}(v)$ and ${\mathcal{F}}(z)$ are squares centered at 0, it suffices to see that the highest norm of an element in ${\mathcal {F}}(v)$ is less than the smallest norm of an element on the border of ${\mathcal{F}}(z)$. As these are reached at $v (\frac{1}{2} +\frac{1} {2} i)$ and $ z (\frac{1}{2})$ respectively, the inclusion is a consequence of the chain of inequalities
$$
\mbox{$
N\left(v(\frac{1}{2}+\frac{1}{2}i)\right)= \frac{1}{2} N(v)< \frac{1}{4}N(z)= N\left(z(\frac{1}{2})\right)$.}
$$
(b) The left hand inclusion holds because the element with the lowest norm on the border of $ {\mathcal {F}}(z)$ has norm $N(z)/4$.
The right hand one is due to the property of the normalized remainder $r$ in the Euclidean division that $N(r)\le N(z)/2$.
\end{proof}

The units (that is, the invertible elements) of $\mathbb{Z}[i]$ are precisely the elements with norm 1, that is $\pm1, \pm i$.  The {\em associates} of a Gaussian integer $z$ are the products $uz$, where $u$ is a unit. Thus, the associates of $z$ are $\pm z, \pm iz$. Clearly two associated numbers $z$ and $uz$ generate the same ideal in $\mathbb{Z}[i]$. So given $z_1,\dots,z_n\in\mathbb{Z}[i]$, the expressions $\mbox{\rm gcd}(z_1,\dots,z_n)$ and $\mbox{\rm lcm}(z_1,\dots,z_n)$ are defined up to associates.
Note that in general $\mathcal{F}(z)\neq\mathcal{F}(uz)$ and thus, for $v\in \mathbb{Z}[i]$, we may have  $v \ (\mbox{\bf mod } z)\neq v \ (\mbox{\bf mod } uz)$. For example, $z/2\in \mathcal{F}(z)$ but $z/2\notin \mathcal{F}(-z)$ and thus, if   $z/2\in\mathbb{Z}[i]$ then we have $z/2 \ (\mbox{\bf mod } z)=z/2,  z/2 \ (\mbox{\bf mod } -z)=-z/2$.

\begin{proposition}
\label{Fstar}
Let $u,v,z$ be three Gaussian Integers such that $u$ is a unit. The following properties hold.\newline
(a) $\mathcal{F}^o(z)=\mathcal{F}^o(uz)$ and $\overline{\mathcal{F}}(z)=\overline{\mathcal{F}}(uz)$. \newline
(b) If $v \ (\mbox{\bf mod } z) \in \mathcal{F}^o(z)$, then $v \ (\mbox{\bf mod } uz)=v \ (\mbox{\bf mod } z)$.
\end{proposition}
\begin{proof}
The proof of (a) is straightforward from the definitions of $\mathcal{F}^o(z)$ and $\overline{\mathcal{F}}(z)$. (b)  If $v=qz+r$ then also  $v=(q/u)uz+r$, so the result  is a direct consequence of (a).
\end{proof}

As a consequence of this result we find that if $v \ (\mbox{\rm mod } \mbox{\rm lcm}(z_1,\dots,z_n))\in \mathcal{F}^o(\mbox{\rm lcm}(z_1,\dots,z_n))$ for some choice of $ \mbox{\rm lcm}(z_1,\dots,z_n)$ and $v \ (\mbox{\rm mod } \mbox{\rm lcm}(z_1,\dots,z_n))$, then the Gaussian Integer given by the expression $v \ (\mbox{\bf mod } \mbox{\rm lcm}(z_1,\dots,z_n))$ is uniquely determined, and moreover  $v \ (\mbox{\bf mod } \mbox{\rm lcm}(z_1,\dots,z_n)) \in \mathcal{F}^o(\mbox{\rm lcm}(z_1,\dots,z_n))$.

\begin{proposition}
\label{v1v2}
Let $v_1,v_2,z$ be Gaussian Integers such that $v_1,v_2\in \overline{\mathcal{F}}(z)$. If $ v_1\equiv v_2 \ (\mbox{\rm mod } z)$ and $v_1\in \mathcal{F}^o(z)$, then $v_1=v_2$.
\end{proposition}
\begin{proof}
 If $ v_1\equiv v_2 \ (\mbox{\rm mod } z)$ with $v_1 \in \mathcal{F}^o(z),v_2\in \overline{\mathcal{F}}(z)$, then there exist a unit $u$ such that $v_1,v_2\in \mathcal{F}(uz)$. The result follows from the uniqueness of the remainder in a fundamental domain stated in Proposition \ref{DivUnica}.
\end{proof}

\subsection{A Mignotte SSS over $\mathbb{Z}[i]$}

Now that we have found a way to guarantee the uniqueness of a remainder, we may expose the desired extension of the scheme.
Let  ${\mathcal{P}}=\{1,\dots,n\}$ be a set of $n$ participants and let  $\mathbf{m}:m_1,m_2,\dots,m_n$,
be a sequence of $n$ (not necessarily pairwise coprime)  non-zero Gaussian Integers.
We assign the Gaussian Integer $m_i$ to each participant $i$. To abbreviate, given a coalition  $C\subseteq {\mathcal{P}}$ we write $\mbox{lcm}(C)=\mbox{lcm}\{ m_i \ : \ i\in C\}$. We say that $N(\mbox{lcm}(C))$ is the norm of  $C$.

Let $m^-,m^+$ be two integers such that $4 m^-<m^+$  and the norm of no coalition lies within the interval $(m^-,m^+)$ , that is, for all $C\subseteq {\mathcal{P}}$ we have either $N(\mbox{lcm}(C)) \le m^-$ or $N(\mbox{lcm}(C)) \ge m^+$. Furthermore, we impose that $\frac{\pi (m^+-4m^-) }{ 4 m^-} > 1$. We may now consider the access structure over $\mathcal{P}$
$$
{\mathcal{A}}=\{ A\subseteq {\mathcal{P}} \ : \ N(\mbox{lcm}(A)) \ge m^+ \}.
$$
Let us now develop a SSS $\mathcal{R}$ realizing the structure $\mathcal{A}$. The set of secrets to be shared is
$$
{\mathcal{S}}=\{ s\in Z[i] \ : \ m^-\le N(s) < \frac{m^+}{4}\}.
$$
The procedure goes as usual. Given a secret $s\in{\mathcal{S}}$,   the shares $s_i$ are $s_i= s \ (\mbox{\bf mod } m_i)$.  If an authorized coalition $A$  wants to recover the secret, they may solve the system of congruence equations
$$
(S_A) \hspace*{1cm} x\equiv s_i \ (\mbox{mod } m_i) \hspace*{1cm} i\in A
$$
whose solution is unique modulo $\mbox{lcm}(A)$. The secret is $x \ (\mbox{{\bf mod}  lcm}(A))$.

Please note that  for any coalition $C$ the corresponding system $(S_C)$ has a solution, since $s \ (\mbox{mod  lcm}(C))$ is so. Thus we only need the Chinese theorem to find the solution and to guarantee its uniqueness.

\begin{theorem}
The above method is correct and gives a secret sharing scheme whose access structure is $\mathcal{A}$.
\end{theorem}
\begin{proof}


An authorized coalition $A$ can solve the system and, from the solution $x$, obtain  $x \ (\mbox{{\bf mod} lcm}(A))\in {\mathcal{F}} (\mbox{lcm}(A))$. Since  $N(s)< m^+/4 \le N(\mbox{lcm }(A))/4$, Proposition \ref{contenciones} guarantees that $s\in \mathcal{F}^o(\mbox{lcm}(A))$, and by Proposition \ref{Fstar} this holds for any choice of $\mbox{lcm}(A)$. 
Then, from Proposition \ref{v1v2}, we have  $x \ (\mbox{{\bf mod}  lcm}(A))=s$, and $A$ successfully recovers the secret.
An unauthorized coalition $B$ may find the solution $x$ of $(S_B)$ modulo $\mbox{lcm}(B)$. However, since $N(s)\ge m^- \ge N(\mbox{lcm}(B))>N(x)$, it holds that $s\neq x$.  $B$ can  also compute
$x \ (\mbox{{\bf mod}  lcm}(B))\in {\mathcal{F}} (\mbox{lcm}(B))$. But since
$N(s)> m^-/2 > N(\mbox{lcm }(B))/2$, again according to Proposition \ref{contenciones},  we have $s\not\in {\mathcal{F}} (\mbox{lcm}(B))$. Hence $x \ (\mbox{{\bf mod}  lcm}(B))\neq s$. We focus on this on section 3.3.
\end{proof}

\begin{example}

We will now see how Example 1 would work under this new method. 
Let $n=3$, $\mathbf{m}: 7+4i, -3-13i, 11+8i$ and take $t=2$. The set of secrets will now be
$\mathcal{S}=\{ s\in \mathbb{Z}[i] : 185 \le N(s) <2892 \}$. Let $s=70-70i$. Note that $N(s)=9800$ so $s$ is not a valid secret, and thus cannot be shared. This explains why the method proposed in \cite{OTS} does not work on this case. 

\end{example}
\subsection{Some Properties}

Once the method is described we may study some of its properties.
Computing the information rate of this scheme leads to the so called {\em Gauss circle problem}, which asks about the number of Gaussian Integers inside a circle of radius $r>0$ centered at the origin. That is to say, the number of Gaussian Integers $z$ such that $N(z)\le N(r)=r^2$, \cite{circulo}. We denote this number by $\mathfrak{N}(r)$. Since, on average, each unit square contains
one Gaussian Integer, $\mathfrak{N}(r)$  is approximately equal to the area of a circle of radius $r$,
$$
\mathfrak{N}(r)\sim \pi r^2
$$

It is also known a explicit expression for this number:

$$
\mathfrak{N} (r) = 1 + 4 \sum_{j = 0}^{\infty} ( \lfloor  \frac{ r^2 }{4j + 1} \rfloor  - \lfloor \frac{r^2}{4j + 3}  \rfloor )
$$
Then the number of possible secrets to be shared is
$$
|{\mathcal{S}}|=\mathfrak{N}(\frac{\sqrt{m^+-1}}{2})-\mathfrak{N}(\sqrt{m^--1})  
$$
$$
= 4 \sum_{j = 0}^{\infty} ( \lfloor  \frac{ m^{+} - 1}{16j + 16} \rfloor  - \lfloor \frac{m^{+} - 1}{16j + 12}  \rfloor ) - \lbrace 4 \sum_{j = 0}^{\infty} ( \lfloor  \frac{ m^{-} - 1}{4j + 1} \rfloor  - \lfloor \frac{m^{-} - 1}{4j + 3}  \rfloor )\rbrace
$$
$$
\sim\pi \left(\frac{m^+-1}{4}-m^-+1 \right)\sim\pi \left(\frac{m^+}{4}-m^- \right).
$$
On the other hand, the set $\mathcal{S}_i$ of possible shares for participant $i$ is precisely the set of congruence classes of Gaussian Integers modulo $m_i$. It is well known that the number of such congruence classes is the norm $N(m_i)$, \cite{Fra}.

Unfortunately, the scheme is not perfect.  An unauthorized coalition $B$ can solve $(S_B)$ and thus it may compute $x=s \ (\mbox{mod lcm}(B))$. So it may deduce that the secret is of the form   $x+\lambda \mbox{ lcm}(B)$ for some $\lambda\in Z [i]$. Thus it can discard all secrets of $\mathcal{S}$ not satisfying this condition. Since
$N(\mbox{\rm lcm}(B))<m^-$, this fact reduces for $B$ the set of secrets to a set of cardinality equal to

$$
\left\lbrace 4 \sum_{j = 0}^{\infty} ( \lfloor  \frac{ m^{+} - 1}{16j + 16} \rfloor  - \lfloor \frac{m^{+} - 1}{16j + 12}  \rfloor ) - \lbrace 4 \sum_{j = 0}^{\infty} ( \lfloor  \frac{ m^{-} - 1}{4j + 1} \rfloor  - \lfloor \frac{m^{-} - 1}{4j + 3}  \rfloor )\rbrace \right\rbrace / N(lcm(B))
$$

$$
\sim \frac{\pi (m^+-4m^-)/4 }{ N(\mbox{\rm lcm}(B))} >
\frac{\pi (m^+-4m^-) }{ 4 m^-}.
$$
So this number should be large enough in order to guarantee the security of the scheme. Note that we have imposed that number to be larger than $2$, so that the coalition $B$ never compute the correct secret.

\begin{example}
\label{ej2}
Let $\mathbf{m}$ be the sequence of pairwise coprime Gaussian Integers $\mathbf{m}:
15+14i, 10-18i, 13+16i$. Take $m^-=425, m^+=178504$.  This choice leads to a $(2,3)$ threshold access structure.
The set of secrets is
$\mathcal{S}=\{ s\in \mathbb{Z}[i] : 425 \le  N(s) \le 44625 \}$. The number of possible secrets is then
$$
|{\mathcal{S}}|=\mathfrak{N}(\sqrt{44625})-\mathfrak{N}(\sqrt{424}) \sim   138858
$$
while the set of possible shares has cardinality at most $N(13+16i)=425$. Of course the scheme is not perfect. As noted before, an unauthorized coalition $B$ can reduce the whole set of secrets to a set of size at least (approximately)
$$
\frac{\pi \left((m^+/4)^2-(m^--1)^2 \right)}{m^- }\approx 327.
$$
\end{example}

\subsection{Any access structure can be  realized by a Mignotte SSS over $\mathbb{Z}[i]$}

An interesting question is to characterize those access structures that can be realized by a Mignotte construction.
The next proposition is analogous to Theorem 2.2 of \cite{GMTodas}. Even the proof is an adaptation of it, we will include it for the convenience of the reader.

\begin{theorem}
\label{todas}
For any access structure $\mathcal{A}$ over $n$ participants and any integer $S$, there is sequence $\mathbf{m}:m_1,\dots, m_n$ of Gaussian Integers such that the Mignotte SSS over $\mathbb{Z}[i]$ arising from $\mathbf{m}$ realizes the structure $\mathcal{A}$ with a set of secrets $\mathcal{S}$ of cardinality $|\mathcal{S}|\ge S$.
\end{theorem}
\begin{proof}
Let $\mathcal{A}$ be an access structure over a set $\mathcal{P}$ of $n$ participants and let $B_1,\dots,B_t$ be the maximal unauthorized coalitions for $\mathcal{A}$. Choose  $t$ pairwise coprime Gaussian Integers $\mu^{(1)},\dots,\mu^{(t)}$ with $N(\mu^{(j)})>8$ for all $j=1,\dots,t$, and let $\mu=\mu^{(1)} \cdots\mu^{(t)}$. Now for $i=1,\dots,n$ and $j=1,\dots,t$, define
$$
{\mu}^{(j)}_i=
\left\{ \begin{array}{ll}
1         & \mbox{if $i\in B_j$} \\
{\mu}^{(j)}     & \mbox{if $i \not\in B_j$}
\end{array} \right.
$$
and  $m_i={\mu}^{(1)}_i\cdots {\mu}^{(t)}_i$. Then the Mignotte scheme associated to the sequence $\mathbf{m}:m_1,\dots, m_n$ realizes the structure $\mathcal{A}$. To see that let
$$
m^+=N({\mu}) \; \mbox{ and } \;
\ell=m^+/\min\{ N({\mu}^{(1)}),\dots, N({\mu}^{(t)})  \}.
$$
Note that $4\ell <m^+$. Let $A\in\mathcal{A}$ be an authorized coalition. For each maximal unauthorized coalition $B_j$ there is a participant $i$ (depending on $j$)  such that $i\in A\setminus B_j$. Thus ${\mu}^{(j)}_i={\mu}^{(j)}$ and hence ${\mu}^{(j)} | \mbox{lcm}(A)$. Since this reasoning holds for any $j=1,\dots,t$, we conclude that
$\mbox{lcm}(A)={\mu}$, so $N(\mbox{lcm}(A))=m^+$. Conversely, let $B$ be an unauthorized coalition. Then there exists $j$ such that $B\subseteq B_j$. It follows that ${\mu}^{(j)}_i=1$ for all $i\in B$, and hence $N(\mbox{lcm}(B))\le m^+/N({\mu}^{(j)})\le \ell$.
We have proved that ${\mathcal{A}}=\{ A\subseteq {\mathcal{P}} \ | \ N(\mbox{lcm}(A)) \ge m^+ \}$, so the sequence $\mathbf{m}:m_1,\dots, m_n$ realizes the structure $\mathcal{A}$ with set of secrets
$\mathcal{S}=\{ s\in Z[i] \ : \ m^-\le N(s) < \frac{m^+}{4}\}$, where $m^-=\max\{ N(\mbox{lcm}(B_1)),\dots,N(\mbox{lcm}(B_t))\}\le \ell$. Since  $\frac{m^+}{4}-m^-\ge \ell$, 
by choosing the Gaussian Integers $\mu^{(1)},\dots,\mu^{(t)}$ having large enough norm, it is clear that we can always get $|\mathcal{S}|\ge S$.
\end{proof}



The simplest case of Mignotte's construction arises when all the numbers in the sequence $\mathbf{m}$ are pairwise coprime. Let us recall that an access structure $\mathcal{A}$ is  {\em weighted threshold} if there is an $n$-tuple $\mathbf{w}=(w_1,\dots,w_n)$ of positive {\em weights} and a threshold $t$ such that $\mathcal{A}=\{ A\subseteq \mathcal{P} : \sum_{i\in A} w_i \ge t\}$. The proof that a weighted access structure defined by real weights can be also written by using integer weights is straightforward.

\begin{proposition}
If the Gaussian Integers in the sequence $\mathbf{m}:m_1,m_2,\cdots,m_n$ are pairwise coprime, then for any  $m^+$ the access structure  $\mathcal{A}(\mathbf{m},m^+)$   is a weighted threshold access structure.
\end{proposition}
\begin{proof}
A coalition $C$ is authorized if and only if $N(\mbox{lcm}(C))=\prod_{i\in C} N(m_i)\ge m^+$, that is if and only if $\sum_{i\in C} \log(N(m_i)) \ge\log(m^+)$, so $\mathcal{A}(\mathbf{m},m^+)$   is weighted threshold structure with weights $ \log(N(m_i))$, $j=1\dots,n$, and threshold  $\log(m^+)$.

\end{proof}

\begin{example}
Let $\mathbf{m}$ be the sequence of pairwise coprime Gaussian Integers $\mathbf{m}:
6+5i, 1-9i,13+16i$, and take $m^-=5002, m^+=25925$.  This sequence leads to the access structure $\mathcal{A}$ whose minimal authorized coalitions are $\{1,3\}$ and $\{2,3\}$, which 
certainly 
is a weighted threshold access structure with weights $1,1,2$ and threshold $t=3$.
\end{example}

\end{document}